\documentclass[11pt,letterpaper]{article}
\pdfoutput=1

\usepackage[dvipsnames,usenames]{xcolor}
\PassOptionsToPackage{hyphens}{url}
\usepackage[colorlinks=true,urlcolor=Blue,citecolor=Green,linkcolor=BrickRed]{hyperref}
\usepackage{amsmath,amsthm,amssymb,amsfonts}
\usepackage[utf8]{inputenc}
\usepackage{todonotes}
\usepackage{xspace}
\usepackage{multicol}
\usepackage{thmtools,thm-restate}
\usepackage{authblk}
\usepackage{multirow}
\usepackage{algorithm}
\usepackage{algorithmicx}
\usepackage[noend]{algpseudocode}
\usepackage{tikz}
\usepackage[noadjust]{cite}
\usepackage{cellspace}

\usepackage[left=1in,right=1in,top=1in,bottom=1in]{geometry}

\newtheorem{theorem}{Theorem}[section]
\newtheorem{lemma}[theorem]{Lemma}
\newtheorem{corollary}[theorem]{Corollary}
\newtheorem{definition}[lemma]{Definition}

\newtheorem{fact}[lemma]{Fact}
\newtheorem{claim}[lemma]{Claim}

\newtheorem{def-restatable}[theorem]{Definition}
\newtheorem{obs-restatable}{Lemma}
\newtheorem{lemma-restatable}[theorem]{Lemma}
\newtheorem{corollary-restatable}[theorem]{Corollary}

\newcommand{\Oh}{\mathcal{O}}
\newcommand{\OO}{O}
\newcommand{\NN}{\Gamma}
\newcommand{\toggle}{\textsf{toggle}}
\newcommand{\UndoLastMove}{UndoAndRemoveLastMoveFrom$S_1$}
\newcommand{\ToggleBlackNode}{ToggleBlackNodeFromVAndAppendTo$S_1$}

\newcommand{\Id}{Id}
\newcommand{\OV}{OV}

\newcommand{\FIGURE}[4]{
\begin{figure}[#1]
\begin{centering}
\includegraphics[width={#2}\textwidth]{{#3}.pdf}
\caption{#4}
\label{fig:#3}
\end{centering}
\end{figure}
}

\title{Sorting Signed Permutations by Reversals in Nearly-Linear Time\footnote{This work was partially funded by the grant ANR-20-CE48-0001 from the French National Research Agency (ANR).}}
\date{}
\author[1]{Bartłomiej Dudek}

\author[1]{Paweł Gawrychowski}

\author[2]{Tatiana Starikovskaya}

\affil[1]{Institute of Computer Science, University of Wrocław, Poland}
  \affil[\ ]{\texttt{\{bartlomiej.dudek,gawry\}@cs.uni.wroc.pl}}
\affil[2]{DIENS, \'{E}cole normale sup\'{e}rieure, PSL Research University, France}
\affil[\ ]{\texttt{tat.starikovskaya@gmail.com}}

\begin{document}
\maketitle

\begin{abstract}
Given a signed permutation on $n$ elements, we need to sort it with the fewest reversals.
This is a fundamental algorithmic problem motivated by applications in comparative genomics,
as it allows to accurately model rearrangements in small genomes.
The first polynomial-time algorithm was given in the foundational work of Hannenhalli and Pevzner [J. ACM'99].
Their approach was later streamlined and simplified by Kaplan, Shamir, and Tarjan [SIAM J. Comput.'99] and
their framework has eventually led to an algorithm that works in $\Oh(n^{3/2}\sqrt{\log n})$ time given by
Tannier, Bergeron, and Sagot [Discr. Appl. Math.'07]. However, the challenge of finding a nearly-linear time algorithm remained
unresolved. In this paper, we show how to leverage the results on dynamic graph connectivity
to obtain a surprisingly simple $\Oh(n \log^2 n / \log \log n)$ time algorithm for this problem.
\end{abstract}

\section{Introduction}

The main goal of comparative genomics is to design efficient methods for comparing genomes of different species.
A genome can be thought of as a linear sequence of genes, identified by numbers, with every gene having a direction,
or a sign. Already in the eighties, it has been observed that seemingly very different genomes might actually
have a small reversal distance, defined as the smallest number of reversals to transform one genome into another.
We refer to the comprehensive treatment of different aspects of such an approach in the book by Pevzner~\cite{Pevzner00}.
This brings the question of computing this number efficiently.

A clean combinatorial formulation of this problem was first formulated by Day and Sankoff~\cite{DayS87}.
We think that each genome is a permutation $\pi$ on
$\{1,2,\ldots,n\}$ where every element additionally has a sign $+$ or $-$. A reversal takes a segment
of such a $\pi$, reverses it and additionally flips the sign of its elements. Given a signed permutation
$\pi$ and $\pi'$, we want to find a shortest sequence of reversals that transforms $\pi$ into $\pi'$. The
length of such a sequence is called the reversal distance between $\pi$ and $\pi'$. By renaming the elements,
it is enough to consider sorting a given signed permutation $\pi$, that is, transforming it into $(1,2,\ldots,n)$.

The challenge of designing an efficient algorithm for sorting signed permutations by reversals has a long history.
Initially, only approximation algorithms were known~\cite{KececiogluS95,BafnaP96}.
In their famous JACM paper titled ``Transforming Cabbage into Turnip'', Hannenhalli and Pevzner~\cite{HannenhalliP99} were able
to design a polynomial time algorithm for this problem. This is somewhat surprising, as the unsigned version of the
problem is in fact NP-hard~\cite{Caprara97}. Their main contribution was a duality theorem that connects the
reversal distance with some combinatorial parameters, which results in a fairly complicated algorithm running
in $\Oh(n^{4})$ time. Very soon afterwards, Berman and Hannenhalli~\cite{HannenhalliP99} showed how to
improve its running time to $\Oh(n^{2}\alpha(n))$, however the underlying combinatorial theory was still very
complicated. Fortunately, Kaplan, Shamir, and Tarjan~\cite{KaplanST99} were able to considerably simplify
the combinatorial theory and the resulting algorithm, thus obtaining a faster quadratic time algorithm,
and Bergeron~\cite{Bergeron05} simplified it even further, obtaining a different algorithm with cubic running time.
However, no worst-case subquadratic time algorithm was known at this point. Ozery-Flato and Shamir~\cite{Ozery-FlatoS03}
designed a family of permutations on which any algorithm based on a similar approach needs at least quadratic time,
and explicitly raised the question of designing a subquadratic algorithm.

An empirical answer to this question was provided by Kaplan and Verbin~\cite{KaplanV05}, who designed an algorithm
based on a random walk in a certain graph. On a random input, their algorithm seems to find the correct shortest sequence
with high probability in $\Oh(n\sqrt{n\log n})$ time. More importantly, at the heart of their algorithm lies a
structure for maintaining a permutation under reversals in $\Oh(\sqrt{n\log n})$ time per
operation. An extension of this structure was used by Tannier, Bergeron, and Sagot~\cite{TannierBS07}
to design a worst-case $\Oh(n\sqrt{n\log n})$ time algorithm that works for any input.
Their algorithm departs in an interesting way from the previous approaches, which were based on repeatedly
applying the so-called safe reversals (reversals that decrease the distance by~$1$). Instead, they maintain
a sequence of reversals, and repeatedly increase its length by inserting (possibly multiple) reversals somewhere in the
middle (and not just at the end). Their algorithm crucially depend on the observation that the position in which we insert
the new reversals never moves to the right in the current sequence, which allows for an efficient implementation.
Finally, Han~\cite{Han06} showed how to tweak the structure of Kaplan and Verbin~\cite{KaplanV05} to
implement some of its operations in $\Oh(\sqrt{n})$ time. While not fully described in his paper, it is plausible 
that such an approach in fact works for speeding up all the operations.

After the subquadratic time algorithm of Tannier, Bergeron, and Sagot~\cite{TannierBS07},
the next challenge is understanding if the complexity of sorting signed permutations by reversals is $\tilde\Oh(n^{1.5})$,
or perhaps is there a near-linear time algorithm?
On a more applied side, Swenson, Rajan, Lin, and Moret~\cite{SwensonRLM10} identified a data-dependent parameter $k$,
observed that it appears to be a small constant for a random input, and designed an $\Oh(n\log n+kn)$ time algorithm.
However, no worst-case near-linear time algorithm was known.
Interestingly, computing the reversal distance itself can be done in linear time~\cite{BermanH96,BaderMY01}.
However, constructing the corresponding sequence seems more challenging.

\paragraph{Our result. }
We design a near-linear time for sorting a signed permutation with reversals. More specifically, our
algorithm outputs a shortest sequence of reversals sorting a given signed permutation on $n$ elements
in near-linear time.
First, we show how to combine the algorithm of Tannier, Bergeron, and Sagot~\cite{TannierBS07} with
a structure for maintaining a graph under inserting/deleting edges and maintaining a minimum spanning tree.
This immediately gives us an $O(n\log^{4}n)$ time algorithm by plugging in e.g. the structure of Holm et al.~\cite{HolmLT01}.
Next, we show that in fact maintaining any spanning forest is sufficient, which allows for a faster $\Oh(n\log^{2}n)$
time algorithm~\cite{HolmLT01}, or $\Oh(n\log^{2}n/\log\log n)$ by plugging in the fastest deterministic structure~\cite{WulffNielsen}
and some careful bookkeeping.

\paragraph{Our method. }
Our algorithm builds on the work of Tannier, Bergeron, and Sagot~\cite{TannierBS07}. The high-level description of their algorithm
is as follows. They maintain a sequence $S$ of reversals, partitioned into $S_{1}$ and $S_{2}$.
$S_{2}$ will be the suffix of the final sequence, but there might be some further reversals inserted into $S_{1}$.
They maintain the current signed permutation obtained by applying all the reversals in $S_{1}$ to the initial permutation.
They also maintain the set of remaining reversals $V$ that can create a new adjacency, that is,
bring elements $i$ and $i+1$ to the adjacent positions. Then, as long as there exists a reversal in $V$ such that its endpoints
contain elements with different signs in the current permutation, such a reversal is applied, removed from $V$, and appended to $S_{1}$
Otherwise, the last reversal of~$S_{1}$ is undone and prepended to $S_{2}$. The main difficulty
is to detect a reversal in $V$ such that its endpoints contain elements with different signs, given that we keep reversing
different segments of the permutation (recall that this involved flipping the sign of every element in the segment).
Instead of designing a new structure for this problem, we reduce it to the well-known fully dynamic connectivity problem.
More specifically, our graph consists of two paths on nodes $\{0,1,\ldots,n+1\}$ and $\{0',1',\ldots,(n+1)'\}$. We think
that the paths consist of blue edges. Throughout the execution of the algorithm, we maintain the two paths,
and for each $i=0,1,\ldots,(n+1)'$ the node $i$ is on one path while the node $i'$ is on the other path.
The order on the nodes on each path corresponds to the current permutation (disregarding the signs), and the path
starting at node $0$ contains the node $i$ if and only if the element $i$ of the permutation has been flipped an even number
of times. It is easy to see how to maintain such invariants by changing 4 edges per reversal in the permutation.
Next, we also encode every reversal in $V$ as two red edges in the graph. This is done in such a way that
the endpoints of a reversal in $V$ contain elements with different signs if and only if its corresponding red edges
connect different blue paths in the current graph. This does not require any modifications to the graph after a reversal,
except that we need to remove its corresponding red edges when removing a reversal from $V$.
We maintain a fully dynamic connectivity structure for the graph; we stress that, even though we have edges of two colours,
we simply maintain all of them in the structure, disregarding the colours. Now, checking if there exists a reversal in~$V$ such that its
endpoints contain elements with different signs reduces to checking if the whole graph is connected.
Extracting such a reversal requires a bit more work, as we have edges of two colors in our graph, and we need
to find a red edge in the current spanning tree. We show how to implement this efficiently by binary searching
over the unique path connecting nodes $0$ and $0'$ in the spanning tree.

\section{Preliminaries}
In this section, we describe the previous framework and summarise it as a concise interface on signed permutations in Theorem~\ref{thm:interface_on_permutation}. Our result follows from an efficient implementation of this interface presented in Section~\ref{sec:faster}. 

We are mostly using the naming convention and definitions from \cite{TannierBS07}.
A signed permutation $\pi$ on $n$ elements is a sequence where each element of $\{1,2,\ldots,n\}$ appears exactly once and has a sign, $+$ (often omitted) or $-$. By $\pi_i$, we denote the $i$th element of the sequence and by $\pi_i^{-1}$ an index $j$ such that $\pi_j = \pm i$. In the problem of sorting by reversals, we are given a signed permutation~$\pi$ and must find a shortest sequence of reversals transforming $\pi$ into the identity permutation $\Id=(1,2,\ldots,|\pi|)$. A reversal $\rho$ of an interval $[i,j]$ transforms a signed permutation  $\pi = (\pi_1,\pi_2,\ldots,\pi_n)$ into a signed permutation $\pi \cdot \rho=(\pi_1,\ldots,\pi_{i-1},-\pi_j,-\pi_{j-1},\ldots,-\pi_i,\pi_{j+1},\ldots,\pi_n)$. 
The length of a shortest sequence of reversals sorting $\pi$ is denoted by $d(\pi)$.

Consider a signed permutation $\pi$ on $n$ elements. 
Following the convention, we extend the permutation $\pi$ by $+0$ at the beginning and $+(n+1)$ at the end, that is hereafter we consider only permutations $\pi=(+0,\pi_1,\ldots,\pi_n,+(n+1))$. We will never involve $0$ or $n+1$ in a reversal and therefore their signs will never change. 

Recall the definition of the overlap graph of a signed permutation $\pi$, $\OV(\pi)$, introduced by Kaplan, Shamir, and Tarjan~\cite{KaplanST99}. See Figure~\ref{fig:example0}. We associate two points $\pi_i^-$ and $\pi_i^+$ to every $i\in[n]$, and also $0^+$ to $0$ and $(n+1)^-$ to $(n+1)$. These points are linearly ordered in such a way that $\pi_i^-\prec\pi_i^+$ if $\pi_i\geq0$ and $\pi_i^+\prec\pi_i^-$ otherwise, and $\pi_i^x \prec \pi_j^y$ whenever $i<j$ for any combination of $x,y\in\{-,+\}$. We then add $n+1$ arcs in such a way that every point is an endpoint of one arc, namely, an arc $v_i$ connects $i^+$ and $(i+1)^-$. The \emph{overlap graph} $OV(\pi)$ of a permutation $\pi$ is a graph in which the nodes correspond to the $n+1$ arcs $v_i$, and the nodes corresponding to two arcs are connected by an edge if the intervals that the arcs span intersect, but none is contained in the other. A node $v_i$ is \emph{black} if $i$ and $(i+1)$  have different signs in $\pi$ and \emph{white} otherwise. 

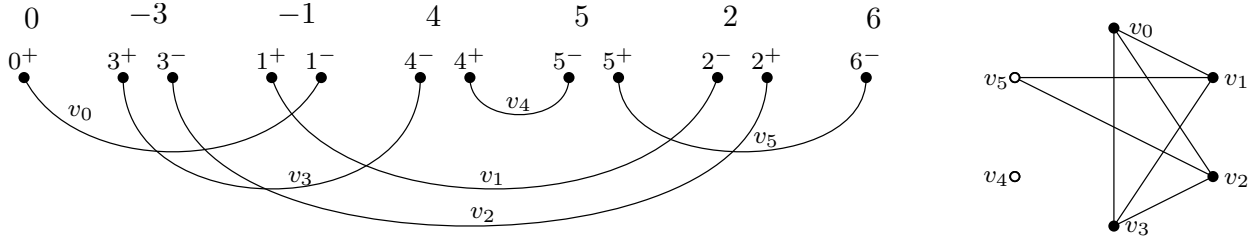
\FIGURE{h}{1}{example0}{Left: permutation $\pi=(0,-3,-1,4,5,2,6)$ with its associated points $\pi^+_i$ and $\pi^-_i$ and arcs~$v_i$. Right: the overlap graph $OV(\pi)$. }

\begin{fact}[{\cite{TannierBS07}}]\label{obs:adjacencies}
 The isolated nodes of $OV(\pi)$ are white and correspond to elements on positions $j$ such that $\pi_j+1=\pi_{j+1}$.
 Such a pair of positions $j,j+1$ is called an \emph{adjacency}.
\end{fact}

A reversal on a permutation corresponds to a natural transformation of its overlap graph, later referred to as \emph{toggle}. Toggles are defined for general graphs where each node is either black or white, not necessarily overlap graphs. Fix such a graph $G$. Let $\NN(v)$ be the neighborhood of a node $v$ of $G$ and $\NN^+(v) = \NN(v) \cup \{v\}$. Define $\toggle(v)$ as a local complementation of $\NN^+(v)$, that is negating the color of every node and complementing the edges in the subgraph of $G$ induced by $\NN^+(v)$ (adding an edge where there was no edge and removing all existing edges). 
The operation is only allowed for black nodes $v$.
Note that after this operation, $v$ always becomes white and isolated.
By $G/v$ we denote the graph obtained from $G$ after applying $\toggle(v)$. We naturally extend this notation to a sequence $S=(s_1,\ldots,s_k)$ of nodes of $G$, namely $G/S:=G/s_1/\ldots/s_k$. 

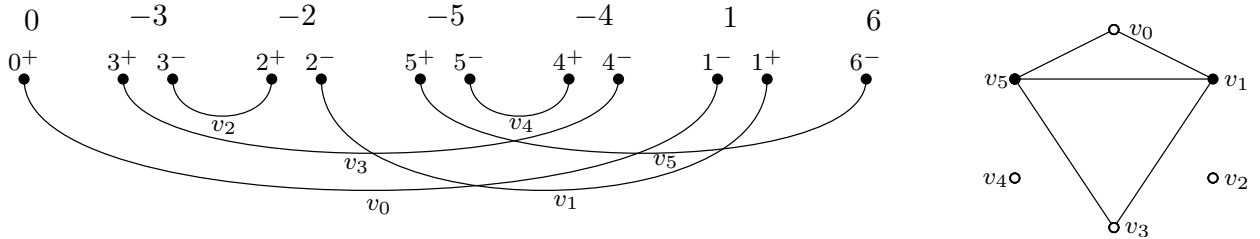
\FIGURE{h}{1}{example1}{Left: permutation $\pi=(0,-3,-1,4,5,2,6)$ from Figure~\ref{fig:example0} after reversal $\rho(v_2)$. Right: the overlap graph $\OV(\pi)/v_2 = \OV(\pi\cdot \rho(v_2))$, see Lemma~\ref{le:toggle_is_reversal}.}

Let $v$ be a black node of $\OV(\pi)$. Denote by $\rho(v)$ a reversal of an interval consisting of the occurrences of all elements $k$ such that both $k^+$ and $k^-$ are spanned by the arc $v$, including its endpoints. Formally, if $v = v_i$, let $a=\min\{\pi^{-1}_i,\pi^{-1}_{i+1}\}$ and $b=\max\{\pi^{-1}_i,\pi^{-1}_{i+1}\}$. If $\pi_a+\pi_b=1$, then $\rho(v)$ is a reversal of $[a,b-1]$, and of $[a+1,b]$ otherwise. See Figure~\ref{fig:example1}.

\begin{lemma}[\cite{TannierBS07,HannenhalliP99,KaplanST99}]\label{le:toggle_is_reversal}
For a permutation $\pi$ and a black node $v$ of $OV(\pi)$, we have that $\OV(\pi)/v=\OV(\pi\cdot\rho(v))$.
\end{lemma}

Lemma~\ref{le:toggle_is_reversal} implies that it suffices to find a shortest sequence of toggles that makes all nodes of $\OV(\pi)$ white and isolated (extending $\pi$ by $0$ and $n+1$ at the ends guarantees that the resulting graph is $\OV(\Id)$). 

A difficulty on the way to computing such a sequence is that $\OV(\pi)$ might have a non-singleton connected component containing only white nodes (``all-white connected component'') that we will not be able to change by toggling. Fortunately, there is a way to overcome this hurdle:

\begin{theorem}[\cite{BermanH96,KaplanST99,BaderMY01}]\label{th:initial_treatment}
Given a signed permutation $\pi$, it is possible to find in $\Oh(n)$ time a sequence of~$t$ reversals such that when we apply them to $\pi$ we obtain a permutation $\pi'$ for which $OV(\pi')$ does not have non-singleton all-white connected components and $d(\pi')=d(\pi)-t$.
\end{theorem}


To compute a shortest sequence of toggles efficiently, we will need one more definition:

\begin{definition}[\cite{TannierBS07}]
A black node $v$ of a graph $G$ is called \emph{safe} if there is no non-singleton all-white connected component in $G/v$.
\end{definition}

Hannenhalli and Pevzner~\cite{HannenhalliP99} showed that instead of looking for a shortest sequence of toggles making all nodes in $\OV(\pi)$ white and isolated, it suffices to find any sequence of toggles of safe nodes: 

\begin{theorem}[\cite{TannierBS07,HannenhalliP99}]\label{thm:safe_decreases_d}
 If $v$ is a safe node of $\OV(\pi)$, then $d(\pi\cdot\rho(v))=d(\pi)-1$.
\end{theorem}

Tannier et al.~\cite{TannierBS07} showed an algorithm that receives a graph $G$ with no non-singleton all-white connected components, and
constructs a sequence $S = (v_1, v_2, \ldots, v_k)$ such that each node $v_{i+1}$ is safe in $G/v_1/\ldots/v_i$, and every node of $G/S$ is
white and isolated. The algorithm proceeds in iterations while maintaining a sequence of toggles partitioned into $S_{1}$ and $S_{2}$ together
with $G/S_{1}$, terminating when there are no black nodes in $G/S_{1}$.
Every iteration consists of two phases. In the first phase, the algorithm repeatedly toggles black nodes in $G/S_{1}$,
appending each toggled node to~$S_{1}$. Then, if $S_{2}$ is nonempty and its first element is not a black node in $G/S_{1}$, the last element
of $S_{1}$ is removed and its corresponding toggle undone. In the second phase, as long as there are no black nodes in $G/S_{1}$
the last element of $S_{1}$ is removed from $S_{1}$ and prepended to $S_{2}$ and its corresponding toggle undone.
See Algorithm~\ref{alg:TBS_orig}.
There, $()$ is the empty sequence, $S_1,S_2$ denotes the concatenation of two sequences and $S,v$ or $v,S$ should be read as $S,(v)$ or $(v),S$ respectively. $S[1]$ is the first element of $S$ and $S[2\ldots]$ is the sequence containing all elements of $S$ but the first one. Correctness of the algorithm is guaranteed by the following theorem:

\begin{theorem}[\cite{TannierBS07}]\label{thm:Tannier_alg_interface}
 Consider a graph $G$ with $n$ black or white nodes with no non-singleton all-white connected components.
 A sequence $v_1,\ldots,v_k$ of nodes in $G$ such that $v_i$ is a black node in $G/v_1/\ldots/v_{i-1}$ for $1\leq i\leq k$ and $G/v_1/\ldots/v_k$ has only white isolated nodes can be found by an algorithm that maintains a sequence $S$ of nodes, the graph $G/S$ after applying the sequence of
 toggles, a subset $V$ of its nodes and performs the following operations:
 \begin{enumerate}
  \item find a black node in $V$,
  \item $\toggle(v)$ for a given black node $v$,
  \item undo the last operation $\toggle(v)$, where $v$ is the last element from $S$ and remove $v$ from the end of $S$,
  \item remove from $V$ all nodes that became isolated and white after the last toggle operation.
 \end{enumerate}
Initially, the set $V$ consists of all non-isolated nodes from $G$.
Every node is toggled at most once (and possibly this is once undone) and each of the above operations is performed in total $\Oh(n)$ times.
Apart from supporting the above operations, the algorithm takes $\Oh(n)$ time. 
\end{theorem}

Unfortunately, \cite{TannierBS07} does not provide a complete proof of the above theorem. More specifically, it first shows in Theorem 3
that any maximal but not total valid sequence $S$ for $G$ can be split into $S=S_{1}, S_{2}$ so that we can find a nonempty
sequence $S'$ for which $S_{1},S',S_{2}$ is a valid sequence for $G$. Next, it states Algorithm~\ref{alg:TBS_orig}, but without any proof
of correctness. At least two issues need to be clarified. First, the algorithm needs to maintain a maximal sequence, so the proof of Theorem 3
needs to be modified so that given a maximal sequence it outputs a longer maximal sequence. Second, the efficiency of the algorithm
hinges on the observation that the position at which we split $S$ into $S_{1}$ and $S_{2}$ never moves to the right, so this needs to be
proven.
As the algorithm serves as a foundation of our improvement, we felt obliged to provide a full proof, see  Appendix~\ref{se:proof_Tannier_alg}.
We stress that the algorithm itself is exactly the same, and the main ideas of the proof were already present in the original paper.

\begin{algorithm}[h]
\begin{algorithmic}[1]
\Function{Process}{graph $G$ with no non-singleton all-white components}
\State $V$ is the set of all nodes of $G$ that are non-isolated or black
\State $S_1,S_2:=()$
\Statex
\Function{\UndoLastMove}{}()
\State  $w:=$ last element of $S_1$
\State  remove the last element (that is: $w$) from $S_1$
\State  undo $\toggle(w)$
\State  \Return w
\EndFunction 
\Statex

\While{there is a black node $v$ in $V$} \label{li:TBS_orig_while}
\State  $\toggle(v)$
\State  remove from $V$ all nodes that became isolated and white (in particular: node $v$)
\State  $S_1:= S_1,v$
\EndWhile
\If{$S_2\ne()$ and $S_2[1]$ is white}
\State \textsc{\UndoLastMove}()
\EndIf
\If{$V$ is empty}
\State \Return $S_1,S_2$
\EndIf
\While{there is no black node in $V$}
\State $w:=$\textsc{\UndoLastMove}()
\State $S_2:=w,S_2$
\EndWhile
\State go to line \ref{li:TBS_orig_while}

\EndFunction
\end{algorithmic}
\caption{Algorithm of Tannier, Bergeron, and Sagot~\cite{TannierBS07}}
\label{alg:TBS_orig}
\end{algorithm}

Now we utilize all the above observations and show that we can efficiently sort by reversals using only a small set of operations on the permutation.
This resembles but extends the interface provided by~\cite{Han06}.

\begin{theorem}\label{thm:interface_on_permutation}
 For every signed permutation $\pi$ of $n$ elements there exists an algorithm that finds a shortest sequence of reversals sorting $\pi$. The algorithm maintains $\pi$ and a set $V\subseteq \{1,2,\ldots, n\}$ under the following operations:
 \begin{enumerate}
  \item\label{item:query_for_pi} query for $\pi_i$ or $\pi^{-1}_i$,
  \item\label{item:get_oriented_i} find $i\in V$ such that $i$ and $i+1$ have different signs in $\pi$, 
  \item\label{item:apply_reversal} apply to $\pi$ a signed reversal of a given interval,
  \item \label{item:remove} remove an element from $V$.
 \end{enumerate}
The algorithm performs $\Oh(n)$ such operations and additionally takes $\Oh(n)$ time.
\end{theorem}
\begin{proof}
We start with extending the permutation $\pi$ with $0$ and $n+1$ at the ends. If $\OV(\pi)$ contains non-singleton all-white connected components, we first compute in $\Oh(n)$ time the sequence of reversals from Theorem~\ref{th:initial_treatment} and apply them to $\pi$. From now on, we assume that every all-white connected components of $\OV(\pi)$ is a singleton. 

Next, we simulate the algorithm from Theorem~\ref{thm:Tannier_alg_interface} on $\OV(\pi)$ based on Lemma~\ref{le:toggle_is_reversal}.
By Fact~\ref{obs:adjacencies}, we initialize $V=\{\max\{\pi_{j},-\pi_{j+1}\}\text{ for } j: \pi_j+1 \neq \pi_{j+1}\}$.
By definition, a black node in $\OV(\pi)$ corresponds to $i$ such that $i$ and $i+1$ have different signs in $\pi$, so we can find it using operation~\ref{item:get_oriented_i} on~$\pi$.
By Lemma~\ref{le:toggle_is_reversal}, toggling a black node $v$ from $OV(\pi)$ corresponds to a reversal of the interval~$\rho(v)$. To undo the last toggle $\toggle(v)$, it suffices to reverse $\rho(v)$ one more time as reversing an interval twice is the identity. However, we cannot retrieve the interval $\rho(v)$ solely from the permutation, because we do not know one of the endpoints of the interval.
Then for every node~$v$ in the sequence~$S$ we store the interval $\rho(v)$ explicitly.

Finally, after a reversal of an interval in $\pi$ we can obtain at most two new adjacencies in $\pi$, at both ends of the interval, because pairs of positions fully inside or outside the interval cannot become or stop being an adjacency.
Hence we can query elements adjacent to the endpoints of the reversed interval, check if an adjacency appeared and, if yes, remove the corresponding element from~$V$.
Hence, we can implement operations 1-4 from Theorem~\ref{thm:Tannier_alg_interface} using operations 1-4 from Theorem~\ref{thm:interface_on_permutation}.

When the algorithm terminates, $\OV(\pi)$ consists only of white isolated nodes. Furthermore, by construction, all chosen nodes were safe (otherwise we would have obtained an all-white connected component), so due to Theorem~\ref{thm:safe_decreases_d} the sequence is the shortest possible.
\end{proof}

\section{Improved algorithm for sorting signed permutations}
\label{sec:faster}
In this section, we show how to efficiently implement all the four operations required in the interface presented in Theorem~\ref{thm:interface_on_permutation}. Recall that $\pi$ is a signed permutation on~$n$ elements and augment it with $0$ and $n+1$ at the ends that do not take part in any reversal.
\begin{fact}[\cite{KaplanV05}, Theorem 5]\label{fact:bst_for_maintaining_pi}
 There exists algorithm implementing operations~\ref{item:query_for_pi} (query for $\pi_i$ or~$\pi_i^{-1}$) under reversals on $\pi$ using $\Oh(\log n)$ time per query or reversal.
\end{fact}

Now we explain how to implement operations~\ref{item:get_oriented_i}-\ref{item:remove}.
We create the following graph $G(\pi)$ on $\Oh(n)$ nodes, $\Oh(n)$ blue edges and $n+1$ red edges.
There are two groups of nodes: $0,1,2,3,\ldots,(n+1)$ and $0',1',\ldots,(n+1)'$, where we think that both $i$ and $i'$
correspond to the element $i$ of the permutation.
For every $0\leq i\leq n$, we create a blue edge $b_i = \{i,(i+1)\}$ and a blue edge $b'_i = \{i',(i+1)'\}$. Additionally, if $i$ and $i+1$ have the same sign in $\pi$ we create two red edges $r_i = \{i,(i+1)\}$ and $r_i'=\{i',(i+1)'\}$.
Otherwise, if $i$ and $i+1$ have different sign in $\pi$ we create two red edges $r_i=\{i,(i+1)'\}$ and $r_i'=\{i',(i+1)\}$.
See Figure~\ref{fig:blue-red-ex0}. Intuitively, we will use the blue edges to simulate reversals on the permutation (operation~\ref{item:apply_reversal}), and the red edges correspond to the elements of $V$ (operations~\ref{item:get_oriented_i} and~\ref{item:remove}).

A \emph{blue connected component} is a connected component of a subgraph of $G(\pi)$ containing all nodes of $G(\pi)$, but only the blue edges. 

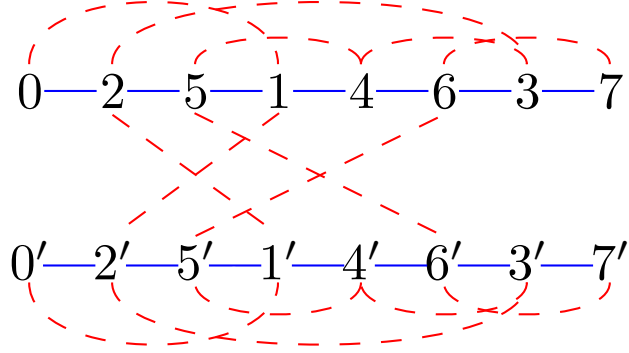
\FIGURE{h}{.5}{blue-red-ex0}{$G(\pi)$ for $\pi=(0,-2,-5,1,-4,6,-3,7)$.
Red edges are dashed and blue edges are solid.}

Let $c$ be the complementation function that satisfies $c(i)=i'$ and $c(i')=i$, for $i\in\{0,\ldots,(n+1)\}$.
 Our algorithm maintains the following invariants:
\begin{enumerate}
 \item There are two blue connected components.
 \item\label{inv:symmetric_components} There is a blue edge $\{a,b\}$ if and only if there is a blue edge $\{c(a),c(b)\}$.
 \item Every blue connected component is a path which, when being read from one of the endpoints ($0$ or $0'$), consists of nodes corresponding to $|\pi_0|, |\pi_1|, \ldots, |\pi_{n+1}|$ in this order.
\end{enumerate}

\noindent
Clearly, the invariants hold before we apply any reversal. To simulate operation~\ref{item:apply_reversal} from Theorem~\ref{thm:interface_on_permutation}, that is, a reversal of an interval $[a,b]$, we perform an operation hereafter referred to as \emph{reconnection} as follows.
We reconnect the four blue edges in such a way that we reverse the nodes at the positions $a,a+1,\ldots,b$ on both paths and ``change their sides''.
See Figure~\ref{fig:reattach} for an example.
Formally, if $b_{a-1} = \{x,y\}$ and $b_{a-1}' = \{u,v\}$, where both $y$ and $v$ correspond to $|\pi(a)|$,
we replace them with $b_{a-1} = \{x,v\}$ and $b_{a-1}' = \{u,y\}$.
Analogously, if $b_b = \{p,t\}$ and $b'_b = \{s,q\}$, where $p$ and~$s$ correspond to $|\pi(b)|$, we replace them with
$b_b = \{p,q\}$ and $b'_b = \{s,t\}$. All other edges remain as they were. See Figure~\ref{fig:blue-red-ex1}.
Clearly, the above invariants are preserved after a reconnection.
We note that translating a reversal into insertion and deletion of edges of $G(\pi)$ can be done efficiently using operation~\ref{item:query_for_pi}.

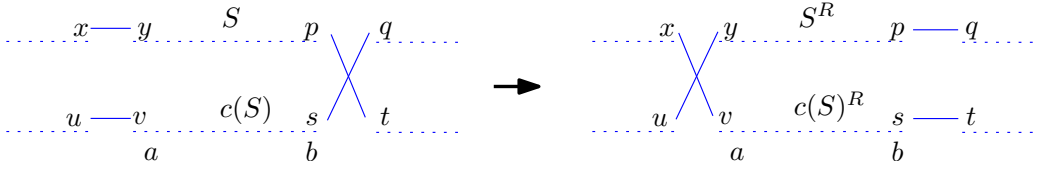
\FIGURE{h}{.85}{reattach}{Reconnection of edges during the reversal of an interval $[a,b]$.
By invariant~\ref{inv:symmetric_components}, $c(x)=u, c(y)=v, c(p)=s$ and $c(q)=t$.
$c(S)$ denotes complementing all nodes in $S$, and $P^R$ denotes reversing the path $P$.
}

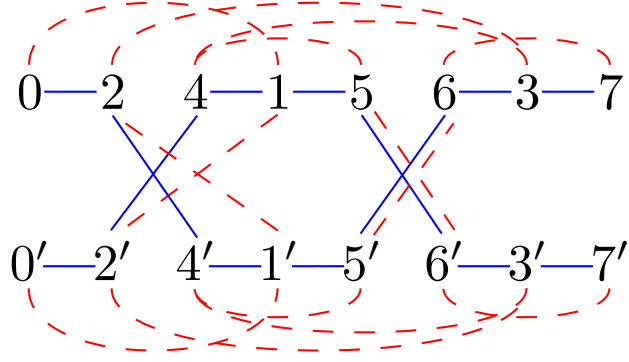
\FIGURE{h}{.5}{blue-red-ex1}{The result of reversing the interval $[2,4]$ on the graph from Figure~\ref{fig:blue-red-ex0}.}

To simulate operation~\ref{item:remove} (removal of an element $i$ from $V$), we remove $r_i$ and $r_i'$ from $G(\pi)$. This operation does not violate the invariants either. 
We finally explain how to simulate operation~\ref{item:get_oriented_i} (find $i \in V$ such that $i$ and $i+1$ have different signs in $\pi$). 
Recall that no reversal includes element~0.
Let $\#(i)$ be the number of performed reversals that changed the sign of $i$.
The blue components fully describe the signs of the elements of $\pi$:

\begin{lemma}\label{le:same_blue}
Node $i$ is in the same blue component as node $0$ if and only if $\#(i)$ is even.
\end{lemma}
\begin{proof}
 By induction on the number of performed reversals.
 By definition of $G(\pi)$, in the beginning every node $i$ is in the same blue component as $0$ and every node $i'$ is in the same component as~$0'$.
 Consider the graph after reversing an interval $[a,b]$.
 Observe that all nodes $j$ and $j'$ such that $\pi^{-1}_j<a$ or $\pi^{-1}_j>b$ are unaffected by the reversal (with regards to being in the same component as $0$), so the claim follows for them by induction.
 On the other hand, nodes $j$ and $j'$ such that $a\leq \pi^{-1}_j\leq b$ change, that is after the reconnection they are in the component of $0$ if and only if they were not before while the parity of $\#(j)$ changes.
 \end{proof}

We immediately derive a property that allows efficient implementation of operation~\ref{item:get_oriented_i} from Theorem~\ref{thm:interface_on_permutation}: 

\begin{corollary}\label{cor:active_red}
 After a sequence of reversals, the endpoints of a red edge $r_i$   are in distinct blue components if and only if $i$ and $i+1$ have different signs in $\pi$. In this case, $r_i$ is called \emph{active}. Symmetrically for a red edge $r_i'$.
\end{corollary}
\begin{proof}
By construction, the claim holds at initialisation. As the red edges are never affected by a reconnection, the claim follows by Lemma~\ref{le:same_blue}. 
\end{proof}

In other words, to implement operation~\ref{item:get_oriented_i} it suffices to find an active red edge in $G(\pi)$ under insertions and deletions of blue edges and deletion of red edges. As a warm-up, we show a very simple $\Oh(\log^4n)$ amortized time algorithm. 

\begin{theorem}\label{th:main_MST}
There exists an algorithm supporting operations~\ref{item:get_oriented_i}-\ref{item:remove} from Theorem~\ref{thm:interface_on_permutation} on a permutation $\pi$ on $n$ elements in $\Oh(\log^4n)$ amortized time.
\end{theorem}
\begin{proof}
We assign weights to the edges of $G(\pi)$ in the following way: all blue edges get weight 0 and for every $i$, the edges $r_i$ and $r_{i}'$ get weight $i+1$. Additionally, we add an edge between $0$ and $0'$ with weight $\infty$.
As the 0-weight edges divide the graph into two components (recall the invariants of the reconnection operation), with the additional edge $\{0,0'\}$
the graph is always connected. Hence, its minimum spanning tree (MST) has total weight $\infty$ when there is no active edge.
Otherwise, namely when there is at least one active red edge, its weight is $w$, where $w-1$ is the smallest $i$ such that $r_i$ is active.
Thus from the weight of the MST we can retrieve an active red edge. By applying an existing approach for dynamic MST that runs in
$\Oh(\log^4n)$ amortized time by Holm et al.~\cite{HolmLT01}, the claim follows.
\end{proof}

Now we show a more efficient approach that does not use dynamic MST.
First we remind some properties of link-cut trees.

\begin{definition}[Link-cut tree~\cite{link/cut}]\label{def:link-cut}
A link-cut tree is a data structure for maintaining a forest (a set of rooted trees) subject to the following operations:
\begin{enumerate}
\item $make\_tree()$: Add a tree consisting of a single node to the forest.
\item $cut(v)$: Cuts a tree containing a non-root node $v$ into two by deleting the edge from the parent of $v$ to $v$;
\item $link(v,w)$: Given two nodes $v,w$ in different trees of the forest, links them by making $v$ a child $w$. The operation assumes that $v$ is the root of the tree it belongs to. 
\item $find\_root(v)$: Finds the root of the tree containing a node $v$.
\end{enumerate}
\end{definition}

Sleator and Tarjan~\cite{link/cut} showed an implementation of link-cut trees with $O(\log n)$ amortized time per operation. Their implementation partitions each tree $T$ in the forest into so-called \emph{preferred} paths, defined as follows: A child $w$ of a node $v$ is the preferred child if the last access (to be described later) within $v$’s subtree was in $w$'s subtree. In particular, if the last access within $v$'s subtree was to $v$, it has no preferred child. A preferred edge is an edge between a preferred child and its parent.
A preferred path is a maximal continuous path of preferred edges in a tree, or a single node if there is no preferred edge incident with it. On top of each preferred path, Sleator and Tarjan maintain a biased 2-3 search tree~\cite{4567825} (a balanced binary search tree), where the nodes are keyed by their depths in $T$.\footnote{Modern variations of link-cut trees represent the paths with splay trees~\cite{splaytrees} to simplify the analysis, but we stick to the original biased 2-3 search trees which are guaranteed to have $\Oh(\log n)$ depth.} 

All operations from Definition~\ref{def:link-cut} are implemented with $access(v)$. If a node $v$ belongs to a tree~$T$ of the forest, the subroutine simulates accessing the nodes in the path from the root of~$T$ to~$v$ by updating the preferred paths. As a corollary, if we access a node $v$ and then a node $w$ of~$T$, the unique path between them will span at most two preferred paths, which will be important later. 

\begin{theorem}\label{thm:main_result}
 There exists an algorithm supporting operations~\ref{item:get_oriented_i}-\ref{item:remove} from Theorem~\ref{thm:interface_on_permutation} on a permutation $\pi$ on $n$ elements in $\Oh(\log^2 n)$ amortized time.
\end{theorem}
\begin{proof}
Consider the graph $G(\pi)$ again. On top of $G(\pi)$, we maintain a data structure for fully-dynamic graph connectivity~\cite{HolmLT01}, which supports edge insertions/deletions in $\Oh(\log^2 n)$ amortized time and maintains a spanning forest $F$ of $G(\pi)$ in a link-cut tree.

There is an active red edge if and only if $0$ and $0'$ are connected by a path, which we can check in $O(\log n)$ time using the link-cut tree.
If $0$ and $0'$ are connected, they belong to the same spanning tree $T \in F$, and are connected by a unique path $P=v_{0}-v_{1}-\ldots -v_{k}$ of $T$,
where $v_{0}=0$ and $v_{k}=0'$.
Since $P$ starts in the blue component of $0$ and ends in the blue component of $0'$, it must contain an active red edge, because
$0$ and $0'$ are in distinct blue components.

We binary search $P$ for an active red edge using the following observation. Consider a fragment $v_{i}-v_{i+1}-\ldots -v_{j}$ of $P$ with the property
that $v_{i}$ belongs to the blue component of $0$ while $v_{j}$ belongs to the blue component of $0'$. Then, consider any $k\in \{i,i+1,\ldots,j-1\}$.
If $v_{k}$ belongs to the blue component of $0'$ then we can recurse on $v_{i}-v_{i+1}-\ldots - v_{k}$. If $v_{k+1}$ belongs to the blue component
of $0$ then we can recurse on $v_{k+1}-v_{k+2}-\ldots -v_{j}$. The remaining case is that $\{v_{k},v_{k+1}\}$ is an active red edge.

To implement the search efficiently, we first call $access(0)$ and $access(0')$ to guarantee that $P$ spans at most two preferred paths.
In fact, as the preferred paths end at $0$ or $0'$, $P$ spans the path to $0$ entirely, and spans a suffix of the path to $0'$.
Thus, $P=v_{0}-v_{1}-\ldots - v_{m} - v_{m+1}- \ldots -v_{k}$, where $v_{0}-v_{1}-\ldots -v_{m}$ is a fragment of one preferred path
and $v_{m+1} - v_{m+2} - \ldots -v_{k}$ is a fragment of the other preferred path. We first use the above observation to narrow down the
search to one of them (or terminate after having found an active red edge). Then, we traverse in the biased 2-3 tree while using the above
observation to decide if we should descend to the left or to the right child. In more detail, we first descend to the lowest common ancestor
of the endpoints of the fragment. Then, we continue the descent, in every step using the above observation twice: if the current node is
$v_{i}$, we apply it on $k=i-1$ and $k=i$. We note that accessing $v_{i-1}$ and $v_{i+1}$ can be done in constant time by maintaining
links to the predecessor/successor of each node on its path.

By Lemma~\ref{le:same_blue}, we can test if a node $v_j$ is in the same blue component as node $0$ by checking the sign of $j$ in $\pi$, which
in turn can be done in $\Oh(\log n)$ time by Fact~\ref{fact:bst_for_maintaining_pi}.
The biased 2-3 trees used to represent each path have depth $\Oh(\log n)$, so the search takes $\Oh(\log^2 n)$ time. 
\end{proof}

We are ready to present the final version of our algorithm. First, we observe that inside Theorem~\ref{thm:main_result} we can use the
faster fully dynamic connectivity structure, which supports edge insertions/deletions in $\Oh(\log^{2}n/\log\log n)$ amortized time~\cite{WulffNielsen}.
The structure can be extended to maintain a spanning forest $F$ of $G(\pi)$ stored in a link-cut tree. Then, when searching for an active
red edge we perform $\Oh(\log n)$ steps. If every step can be implemented in only $\Oh(\log n/\log \log n)$ time then the overall
complexity becomes $\Oh(n\log^{2}n/\log\log n)$. This is indeed possible thanks to the following lemma:

\begin{lemma}\label{lm:number_of_reversals}
For any $\varepsilon > 0$, there is algorithm that maintains a signed permutation $\pi$ on $n$ elements under reversals in $\Oh(\log^{1+\varepsilon} n )$ time and, given an element $i$, retrieves $\#(i)$ in $\Oh(\log n / \log \log n)$ time. 
\end{lemma}
\begin{proof}
Our data structure is a slightly modified data structure from Fact~\ref{fact:bst_for_maintaining_pi} \cite[Theorem 5]{KaplanV05} which we used to query for $\pi_i$ or $\pi_{i}^{-1}$. We start by recalling its main components. It stores the elements of $\pi$ in a balanced binary search tree
allowing for splits and joins, for example a red-black tree. Additionally, at each node it stores the size of its subtree and the reverse flag, which,
if set to true, indicates that the subtree rooted at that node should be read in reverse order, that is from right to left, and the signs of its elements
should be flipped. It maintains the following invariant: the in-order traversal of the tree, taking the reverse flags into account,
must give the permutation. To execute a reversal $\rho(i,j)$, we split the tree into three trees: $T_1$ containing $\pi_1, \ldots, \pi_{i-1}$,
$T_2$ containing $\pi_{i},\ldots, \pi_j$, and $T_3$ containing $\pi_{j+1}, \ldots, \pi_n$. Next, we flip the reverse flag of the root of $T_2$,
and finally join $T_1$ and $T_2$ and then the resulting tree and $T_3$. It is easy to see that the resulting tree satisfies the invariant. Via
a standard approach, the update requires to perform $\Oh(\log n)$ rotations that touch $\Oh(\log n)$ nodes. 

We augment the tree with shortcut pointers that slow down the updates but make the queries faster. For a node~$v$ of depth $d$, the shortcut
points to its ancestor $w$ of depth $\max\{0,d-\lfloor \varepsilon\log \log n \rfloor\}$ and stores the number of reverse flags
on the path from $v$ to $w$ that are set to true. To compute $\#(i)$, we start from the node containing $i$ and follow the shortcuts to the
root aggregating the number of the reverse flags that are set to true in $\Oh(\log n / \varepsilon\log \log n) = \Oh(\log n / \log \log n)$ time,
as the depth of a node in a red-black tree on $n$ elements is bounded by $\Oh(\log n)$. To perform a reversal, we use the algorithm described in the previous paragraph
and then update the shortcuts for all nodes $u$ of the tree such that one of the nodes between $u$ and the endpoint of the shortcut for $u$ was
modified by a rotation. 

We update the shortcuts in $\Oh(\log^{1+\varepsilon} n)$ total time in the following way. Consider a node $x$ modified by a rotation.
Let $T_x$ be the subtree rooted at $x$. We must update the shortcuts for all nodes in $T_x$ that have depth at most $\varepsilon\log \log n$.
The number of such nodes is $\Oh(2^{\varepsilon\log \log n}) = \Oh(\log^{\varepsilon} n)$. The update consists of two steps. First, for each of
the $\varepsilon\log \log n$ lowest ancestors of $x$, we compute the number of reverse flags that are set to true on the path from the ancestor
to $x$ in $O(\varepsilon\log \log n)$ total time. Second, we run a depth-limited DFS from $x$ that traverses all nodes of $T_x$ that have depth
at most $\varepsilon\log \log n$. During the DFS, we store  
the depth of the currently visited node $v$ and the number of reverse flags that are set to true on the path from $x$ to $v$, and update the shortcut
from $v$ by retrieving the appropriate ancestor of~$x$. As the DFS takes time linear in the number of visited nodes, the total time of the two steps
is $\Oh(\log^{\varepsilon} n)$. Hence, the update takes $\Oh(\log^{1+\varepsilon} n)$ across all $\Oh(\log n)$ nodes modified by a rotation.
\end{proof}

By choosing $0 < \varepsilon < 1$ we obtain $\Oh (\log^{1+\varepsilon} ) = \Oh(\log^2 n / \log \log n)$ update time.
Recall that in Theorem~\ref{thm:main_result}, operation~\ref{item:get_oriented_i} required $\Oh(\log n)$ queries about $\#(i)$, so applying the above data structure gives us the following result:

 \begin{corollary}
 There exists an algorithm supporting operations~\ref{item:get_oriented_i}-\ref{item:remove} from Theorem~\ref{thm:interface_on_permutation} on a permutation $\pi$ on $n$ elements in $\Oh(\log^2 n/ \log \log n)$ amortized time.
\end{corollary}

To conclude, by combining the above corollary, Fact~\ref{fact:bst_for_maintaining_pi} and Theorem~\ref{thm:interface_on_permutation} we
obtain an algorithm sorting a signed permutation by reversals in $\Oh(n\log^2 n/ \log \log n)$ time.

\bibliographystyle{plainurl}
\bibliography{biblio}

\appendix

\section{Proof of Theorem~\ref{thm:Tannier_alg_interface}}\label{se:proof_Tannier_alg}
The goal of this section is to establish correctness of  Algorithm~\ref{alg:TBS_orig}.
We will first analyse its slower version, see Algorithm~\ref{alg:TBS_simplified}.
Then, we will consider the faster version, see Algorithm~\ref{alg:TBS_faster}, and
show that the two algorithms are in fact equivalent.
Finally, we will notice that Algorithm~\ref{alg:TBS_faster} is a reformulation of Algorithm~\ref{alg:TBS_orig} that avoids the goto statement.

Consider the following algorithm. We will first show that, given a graph with no non-singleton all-white components,
it construct a sequence of toggles that makes all nodes white and removes all edges from the graph. We note that with every node
$v$ in any sequence we store the neighborhood of $v$ in the graph before toggling $v$. This allows us to undo $\toggle(v)$ in the future.

\begin{algorithm}[h]
\begin{algorithmic}[1]
\Function{Process}{graph $G$ with no non-singleton all-white connected components}
\State $S:=()$
\State $G':=G$
\While{there is a black node $v$ in $G'$}
\State apply $\toggle(v)$ to $G'$
\State $S:= S,v$
\EndWhile
\While{there is a non-singleton all-white connected component in $G/S$}\label{li:main_loop}
\State $U:=$ set of nodes from non-singleton all-white connected components in $G/S$ \label{li:define_U}
\State $S_{1}, S_{2} := S$ for the longest $S_{1}$ s.t. $G/S_1$ has a node of $U$ in a not-all-white component \label{li:splitS}
\State   $G_1:= G/S_1$\label{li:define_g1}
\State   $S_3:= ()$\label{li:set_s3_empty}
\While{there is a black node $v$ from $U$ in $G_1$}\label{li:loop_add_nodes_to_s3}
\State apply $\toggle(v)$ to $G_1$\label{li:toggle_v}
\State $S_3:=S_3,v$ \label{li:loop_add_nodes_to_s3_end}
\EndWhile
\If{$S_2[1]$ is white in $G_1$}\label{li:check}
\State  remove the last element $w$ from $S_3$ \label{li:remove_from_s3}
\State undo $\toggle(w)$ in $G_1$\label{li:undo_toggle_w}
\EndIf
\State $S:=S_1,S_3,S_2$ \label{li:set_new_S}
\EndWhile
\State \Return $S$ 
\EndFunction
\end{algorithmic}
\caption{}
\label{alg:TBS_simplified}
\end{algorithm}

The overall structure of the proof is as follows. 
We start with some technical lemmas concerning the properties of the algorithm.
Then, we establish that in line~\ref{li:splitS} such a split indeed exists, and further $v_{mid}=S_2[1]$ does exist.
Next, we show that in line~\ref{li:remove_from_s3} the sequence $S_{3}$ is non-empty.
In fact, its length is always at least 2, which immediately gives us that the whole algorithm 
terminates in a finite number of steps (in fact, at most linear in the number of nodes).
To prove the correctness, we call a sequence $S$ of toggles \emph{valid} for a graph~$G$ if we can toggle nodes of
$G$ in the order defined by $S$, that is $s_i$ is a black node in $G/s_1/\ldots/s_{i-1}$ for every $1\leq i\leq k$.
First, we will show that in line~\ref{li:main_loop} $S$ is valid for $G$. Next, we will argue that $G/S$ does not contain any black nodes.
Finally, we will prove that the algorithm terminates in a finite number of steps (in fact, at most linear in the number of nodes).
By condition in line~\ref{li:main_loop}, this implies that the found sequence of toggles makes all nodes white and isolated
as required.

Within a single iteration of the main while loop in lines \ref{li:main_loop}-\ref{li:set_new_S} of Algorithm~\ref{alg:TBS_simplified}, let $v_{mid}=S_2[1]$, $\OO$ be the set of nodes in not-all-white components of $G_1/v_{mid}$, and $L = \NN(v_{mid}) \cap \OO$.

\begin{claim}\label{cl:claim1}
Before line~\ref{li:loop_add_nodes_to_s3} of Algorithm~\ref{alg:TBS_simplified}, $U$ consists of all nodes in the non-singleton all-white connected components of $G_1/v_{mid}$. 
\end{claim}
\begin{proof}
Consider a node $v$ in $G_1/v_{mid}$.
If $v$ is white and isolated in $G_1/v_{mid}$, it will be also white and isolated in $G/S$, so it does not belong to $U$.
If $v$ belongs to a white non-singleton component of $G_1/v_{mid}$, its component cannot be affected by any further toggle and remains white in $G/S$, so $v\in U$.
Finally, $v\in\OO$ (a node that belongs to a not-all-white component) cannot belong to $U$ as in this case $S_1, v_{mid}$ satisfies the conditions of the split and is longer than $S_1$, see line~\ref{li:splitS}. 
\end{proof}

Consequently, the nodes of $G_1$ are $v_{mid}$, some white isolated nodes, and the nodes from $\OO$ and~$U$.
Next, we show the following properties of $G_1$ that are preserved after any sequence of toggles of black nodes from~$U$:

\begin{lemma}\label{le:Lemma1}
Consider the graph $G_{1}$ before line~\ref{li:loop_add_nodes_to_s3} of Algorithm~\ref{alg:TBS_simplified},
and any sequence $w_1,w_2,...,w_k$ of nodes from $U$ such that for $1\leq i\leq k$, $w_i$ is a black node in $G_1/w_1/w_2/.../w_{i-1}$. The graph $G_1/w_1/w_2/.../w_k$ satisfies each of the following properties:
\begin{enumerate}
\item[(1)] A node of $U$ is black if and only if it is adjacent to $v_{mid}$.
\item[(2)] There are all possible edges between $\NN(v_{mid})\cap U$ and $L$.
\item[(3a)] There is no edge between the nodes of $U \setminus \NN (v_{mid})$ and the nodes outside $U$.
\item[(3b)] There is no edge between the nodes of $\OO \setminus L$ and the nodes outside $\OO$.
\end{enumerate}
The sets $\OO,U$ and $L$ are defined for $G_1$, whereas adjacency and $\NN(v_{mid})$ refer to $G_1/w_1/w_2/.../w_k$.
\end{lemma}
\begin{proof}

Induction on $k$, following from the properties of the $\toggle$ operation.

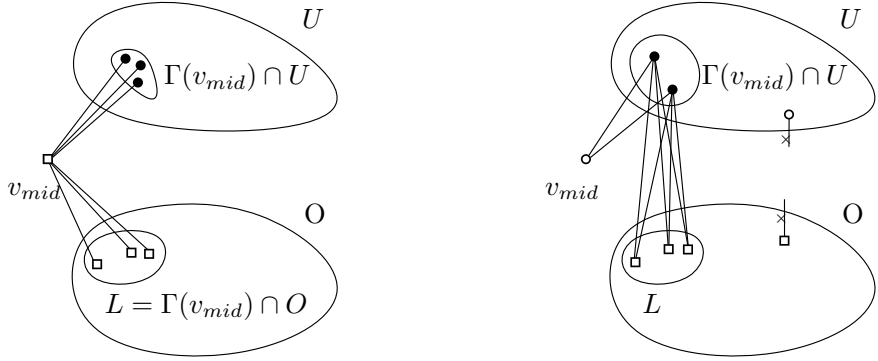
\FIGURE{h}{.7}{vmid}{Full nodes denote black nodes, empty nodes denote white nodes, and square nodes denote nodes which can be either black or white.
$\OO$ is the set of nodes in the not-all-white components of $G_1/v_{mid}$, and $U$ is the set of nodes from the non-singleton all-white connected components in $G/S$.
Left: $v_{mid}$ and the sets $O,L,U$.
Right: invariants (1)-(3b).
$\times$ denotes the edges that are not present in the graph. }

\noindent
\textbf{Base case} ($k=0$). We prove each of the properties separately:
\begin{enumerate}
\item[(1)] Otherwise in $G_1/v_{mid}$ we would have a black node from $U$, which contradicts Claim~\ref{cl:claim1}.
\item[(2)] Otherwise in $G_1/v_{mid}$ we would have an edge from a node $u \in  U$ to a node from $\OO$, which means that $u$ is in a not-all-white component of $G_1/v_{mid}$, which contradicts Claim~\ref{cl:claim1}.
\item[(3a)] If in $G_1$ there was an edge from a node $u \in  U \setminus \NN (v_{mid})$ to a node outside $U$, in $G_1/v_{mid}$ this edge will be preserved and $u$ will be part of a not-all-white component, which contradicts Claim~\ref{cl:claim1}.
\item[(3b)] If in $G_1$ there was an edge from a node $t \in  \OO \setminus L$ to a node $u$ outside $\OO$, then $u\in U$ as $t\notin \NN(v_{mid})$ and hence $u\ne v_{mid}$. 
In $G_1/v_{mid}$ this edge will be preserved and $u$ will be part of a not-all-white component, which contradicts Claim~\ref{cl:claim1}.
\end{enumerate}
\noindent
\textbf{Induction step} ($k-1 \rightarrow k$).
Let $G'=G_1/w_1/w_2/.../w_{k-1}$ and $G''=G'/w_k$.
By definition of $w_k$ and (1) for~$G'$, the node $w_k$ is a neighbor of $v_{mid}$ in~$G'$.
By (2) for $G'$, the node $w_k$ is connected to all nodes from $L$ in~$G'$.

\begin{enumerate}
\item[(1)] Nodes not adjacent to $w_k$ in $G'$ are not affected by $\toggle(w_k)$ so the claim follows for them by induction.
Consider a neighbor $v$ of $w_k$.
If $v\in \NN(v_{mid})$ then by (1) $v$ is black. Hence $v$ is white in $G''$ and $v\notin \NN(v_{mid})$.
If $v\notin \NN(v_{mid})$ then by (1) $v$ is white. Hence $v$ is black in $G''$ and $v\in \NN(v_{mid})$.
\item[(2)]  There are two cases in which $v\in \NN(v_{mid})$ in $G''=G'/v_{mid}$:
\begin{itemize}
 \item $v\in \NN(w_k), v\notin \NN(v_{mid})$: by (3a), in $G'$ the node $v$ is connected to none of the nodes from $\OO$. As $w_k$ is connected to $v_{mid}$ and all nodes from $L$, in $G''$ the node $v\in \NN(v_{mid})$ and is connected to all nodes from $L$, so (2) holds.
 \item $v\notin \NN(w_k), v\in \NN(v_{mid})$: $v$ is not affected by $\toggle(w_k)$ in $G'$, so $v\in \NN(v_{mid})$ in $G''$ and (2) follows by induction.
\end{itemize}
\item[(3a)]  There are two cases in which $v\notin \NN(v_{mid})$ in $G''=G'/v_{mid}$:
\begin{itemize}
 \item $v\in \NN(w_k)$, $v\in \NN(v_{mid})$: in $G'$ the node $v$ is connected to all nodes from $L$ by (2), but no node from $\OO\setminus L$ by (3b). Hence $v\notin \NN(v_{mid})$ in $G''$ and $v$ is not connected to any node from~$\OO$, so (3a) holds.
 \item $v\notin \NN(w_k), v\notin \NN(v_{mid})$: $v$ is not affected by $\toggle(w_k)$, so in $G''$ the node $v\notin \NN(v_{mid})$ and (3a) follows by induction.
\end{itemize}

\item[(3b)] By (3b) for $G'$, nodes from $\OO\setminus L$ are not neighbors of $w_k$ in $G'$, so they are not affected by $\toggle(w_k)$ and (3b) holds also for $G''$. \qedhere
\end{enumerate}
\end{proof}

Now note that after passing the check in line~\ref{li:main_loop} of Algorithm~\ref{alg:TBS_simplified}, the set $U$ is non-empty.
In line~\ref{li:splitS}, we can always find such a split, because in particular $S_1=(), S_2=S$ is a valid split, as $G$ does not have non-singleton white components. Furthermore, by definition of the split, in line~\ref{li:check}  we have $S_2\ne()$, so $v_{mid} = S_{2}[1]$ does exist.
Moreover, by maximality of $S_1$:
\begin{claim}\label{cl:s2_has_only_o}
$S_2[2\ldots]$ contains only nodes from~$\OO$. 
\end{claim}
\begin{proof}
Assume otherwise, and let $S[i]$ be the last node in $S_{2}[2\ldots]$ that does not belong to $\OO$.
Then it must belong to $U$, as every other node is white already in $G_{1}/v_{mid}$. Thus, we could have split
$S$ into $S_{1},v_{mid},S_{2}[2\ldots (i-1)]$ and $S[i \ldots]$.
\end{proof}
Now we show that Algorithm~\ref{alg:TBS_simplified} is well-defined and terminates.

\begin{lemma}\label{le:lemma2}
After the while loop in lines \ref{li:loop_add_nodes_to_s3}-\ref{li:loop_add_nodes_to_s3_end} of Algorithm~\ref{alg:TBS_simplified}, $S_3$ contains at least 2 nodes. 
\end{lemma}
\begin{proof}
If $|S_3|=0$ and there are no black nodes from $U$ in $G_1$, by Lemma~\ref{le:Lemma1}(1) there are no neighbors of $v_{mid}$ in $U$ and by (3) there are no edges from $U$ to $\OO$.
Therefore, the connected components of $U$ are not affected by $\toggle(G_1,v_{mid})$ and $G_1$ has a node of $U$ in a not-all-white component if and only if $G_1/v_{mid}$ has a node of $U$ in a not-all-white component, which contradicts the maximality of $S_1$ in line~\ref{li:splitS}.

If $|S_3|=1$, there is a black node $w\in U$ such that $G_1/w$ has no black nodes in $U$.
By Lemma~\ref{le:Lemma1}(1), the node $w\in \NN(v_{mid})$ and $\NN^+(w)\cap U = \NN(v_{mid})\cap U$ as there are no black nodes in $U$ in $G_1/w$.
Next, by Lemma~\ref{le:Lemma1}(2), $w$ is adjacent to all nodes from $L$ and by Lemma~\ref{le:Lemma1}(3b) adjacent to no nodes from $\OO\setminus L$. Therefore, in $G_1/v_{mid}$ the node $w$ is white and isolated which is the case also for $G_1/S_2=G/S$. Hence $w\notin U$, a contradiction.
\end{proof}

In particular, in line~\ref{li:remove_from_s3} we have $S_{3} \ne()$, so we can indeed remove the last element from $S_{3}$.
Further, in line~\ref{li:set_new_S} we have $|S_1,S_3,S_2|>|S|$, so each iteration of the main while loop results in a longer $S$.
This implies that the algorithm terminates, as each element of $S$ makes at least one node of $G/S$ isolated.

To show the next property, let $W = L \cup \{v_{mid}\}$ and let $G|_Y$ denote the subgraph of graph $G$ induced by a set~$Y$ of nodes.

\begin{claim}\label{cl:Claim2}
Before line~\ref{li:check} of Algorithm~\ref{alg:TBS_simplified}, the subgraph $(G_1/S_3)|_W$ equals $G_1|_W$ if $|S_3|$ is even.
If $|S_3|$ is odd, the subgraph of $(G_1/S_3)|_W$ is the local complement of $G_1|_W$. 
\end{claim}
\begin{proof}
By induction on $|S_3|$.
By Lemma~\ref{le:Lemma1}, every $w_i\in S_3$ is adjacent to $v_{mid}$ and all nodes from~$L$. Therefore, every operation $\toggle(w_i)$ complements the subgraph of the current graph induced by $W$ and the claim follows. 
\end{proof}

We note that before we remove the last node $w_{last}$ in line~\ref{li:remove_from_s3} of the algorithm, $w_{last}$ has exactly the same neighborhood as $v_{mid}$, so removing it and next toggling $v_{mid}$ gives exactly the same result:

\begin{claim}\label{cl:neighborhood_of_vmid}
Consider the graph $G_{1}$ before line~\ref{li:loop_add_nodes_to_s3} of Algorithm~\ref{alg:TBS_simplified},
and let $S_3,w_{last}$ be a valid sequence of nodes from $U$ for $G_1$ such that there is no black node from $U$ in $G_1/S_3/w_{last}$.
If $|S_3|$ is even, in $G_1/S_3$ we have $\NN^+(w_{last})=\NN^+(v_{mid})$.
If $|S_3|$ is odd, $v_{mid}$ has no neighbors from $U$ in $G_1/S_3/w_{last}$.
\end{claim}
\begin{proof}
By definition, there is no black node from $U$ in $G_1/S_3/w_{last}$, so by Lemma~\ref{le:Lemma1}(1) $v_{mid}$ has no neighbors from $U$ in $G_1/S_3/w_{last}$ and the claim follows for odd $|S_3|$.

If $|S_3|$ is even, by Claim~\ref{cl:Claim2}, $v_{mid}$ is adjacent to all nodes from $L$ and is black (because it is black in $G_1$). By Lemma~\ref{le:Lemma1}, $\NN^+(w_{last})\cap \OO = \NN^+(v_{mid})\cap \OO =L$.
Recall that $w_{last}$ is adjacent to~$v_{mid}$. Consider a node $v\in U$.
If $v\in \NN^+(w_{last})$, then it is black in $G_1/S_3$, because in $G_1/S_3/w_{last}$ there is no black node from $U$. Therefore, $v\in \NN^+(v_{mid})$, by Lemma~\ref{le:Lemma1}(1).
If $v\notin \NN^+(w_{last})$, then it is white in $G_1/S_3$, because in $G_1/S_3/w_{last}$ there is no black node from $U$, so $v\notin \NN^+(v_{mid})$, by Lemma~\ref{le:Lemma1}(1).
\end{proof}

\begin{lemma}
In line~\ref{li:set_new_S} of Algorithm~\ref{alg:TBS_simplified}, $S_1,S_3,S_2$ is a valid sequence for $G$.
\end{lemma}
\begin{proof}
By definition of the split $S=S_1,S_2$ and Lemma~\ref{le:Lemma1}, toggling the nodes $S_2[1\ldots]$ affects only the subgraph of $(G_1/v_{mid})|_\OO$, because there are no edges from $\OO$ to outside $\OO$ and all other nodes are either isolated or in all-white connected components.

Clearly, $S_3$ is a valid sequence for $G_1$.
As $v_{mid}$ is black in $G_1$ (toggling it is the first operation in $S_2$), by Claim~\ref{cl:Claim2} the subgraph $(G_1/S_3)|_{W} = G_1|_{W}$ if and only if $v_{mid}$ is black in $G_1/S_3$.
In line~\ref{li:check} of Algorithm~\ref{alg:TBS_simplified}, we ensure that $v_{mid}$ is black in $G_1/S_3$, so then $S_3,v_{mid}$ is a valid sequence for $G_1$.
As the subgraph $G_1|_{\OO\setminus L}$ is not affected by the operations from $S_3$, we obtain that $(G_1/S_3)|_{\OO\cup\{v_{mid}\}} = G_1|_{\OO\cup\{v_{mid}\}}$.

To sum up, the subgraph $(G_1/v_{mid})|_\OO$ is the same as the subgraph $(G_1/S_3/v_{mid})|_\OO$.
As toggling the nodes $S_2[1\ldots]$ affects only the subgraph $(G_1/v_{mid})|_\OO$, we obtain that $S_3,S_2$ is a valid sequence for $G_1$ and the claim follows.
\end{proof}

Let $S_3'$ be the sequence $S_3$ when reaching line~\ref{li:check} and $S''$ be the sequence $S_3$ when reaching line~\ref{li:set_new_S}.

\begin{lemma}\label{le:properties}
Consider the graph $G_{1}$ before line~\ref{li:loop_add_nodes_to_s3} of Algorithm~\ref{alg:TBS_simplified}.
  The following properties hold:
\begin{enumerate}
\item $G_1/v_{mid}|_\OO=G_1/S_3''/v_{mid}|_\OO$.\label{it:prop_subgr_o}
 \item In $G_1/S_3''/v_{mid}$, the nodes from $S_2[2\ldots]$ do not have any neighbors in $U$.\label{it:prop_s2_not_on_u}
\item $G_1/S_3'|_U=G_1/S_3''/S_2|_U$. \label{it:prop_subgr_u}
\end{enumerate}
\end{lemma}
\begin{proof}
From Lemma~\ref{le:Lemma1} every toggled node $v\in S_3''$ is incident to all nodes from $L$ and no nodes from $\OO\setminus L$. 
By Claims~\ref{cl:Claim2} and \ref{cl:neighborhood_of_vmid}, property (\ref{it:prop_subgr_o}) follows.

By Lemma~\ref{le:Lemma1}, in $G_1/S_3'$ there is no edge from $U$ to $\OO$, so by Claim~\ref{cl:neighborhood_of_vmid} this is also the case for $G_1/S_3''/v_{mid}$.
As every element of $S_2[2\ldots]$ is in $\OO$ by Claim~\ref{cl:s2_has_only_o}, property~(\ref{it:prop_s2_not_on_u}) follows. Finally, property~(\ref{it:prop_subgr_u}) follows immediately from property~(\ref{it:prop_s2_not_on_u}) and Claim~\ref{cl:neighborhood_of_vmid}.
\end{proof}

This immediately gives us that in line~\ref{li:set_new_S} of Algorithm~\ref{alg:TBS_simplified}, the sequence $S_3,v_{mid},S_2[2\ldots]$ is valid for $G_1$, so
$S_1,S_3,S_2$ is valid for $G$.
Thus by induction $S$ is valid for $G$ in line~\ref{li:main_loop} of Algorithm~\ref{alg:TBS_simplified}.
Further, we can show by induction that there is no black node in $G/S$.

\begin{lemma}
 There are no black nodes in $G/S$.
\end{lemma}
\begin{proof}
The proof is by induction on the number of iterations of the main loop in lines~\ref{li:main_loop}-\ref{li:set_new_S} of Algorithm~\ref{alg:TBS_simplified}.
Before line~\ref{li:main_loop} the claim is true because we have kept toggling black nodes as long as possible. 

Consider a single iteration of the main loop and graph $G_1$ before line~\ref{li:loop_add_nodes_to_s3}.
Let $G_2=G_1/S_3''/S_2$.
After line~\ref{li:loop_add_nodes_to_s3_end}, there is no black node from $U$ in $G_1/S_3'$, so by Lemma~\ref{le:properties}(\ref{it:prop_subgr_u}), there is no black node from $U$ in $G_2$.
Clearly $v_{mid}$ is white in $G_2$, because $v_{mid}=S_2[1]$ is toggled.
Now we need to show that there is no black node from $\OO$ in $G_2$.
By Lemma~\ref{le:properties}(\ref{it:prop_subgr_o}) and the fact that $S_2[2\ldots]$ contains only nodes from~$\OO$ (Claim~\ref{cl:s2_has_only_o}) we have $G_2|_\OO=G_1/S_3''/v_{mid}/S_2[2\ldots]|_\OO=G_1/v_{mid}/S_2[2\ldots]|_\OO=G_1/S_2|_\OO=G/S|_\OO$ which by induction has no black node.
Hence $G_2$ has no black node in $\OO$ and the claim follows.
\end{proof}

This concludes the proof of the correctness of Algorithm~\ref{alg:TBS_simplified}.
The bottleneck of Algorithm~\ref{alg:TBS_simplified} is that in every iteration of the loop in lines \ref{li:main_loop}-\ref{li:set_new_S}, we need to find the split $S=S_1,S_2$. To implement Algorithm~\ref{alg:TBS_simplified} efficiently, we first establish that
the sets $U$ and the suffixes $S_2$ are monotone in subsequent iterations of the algorithm.
Let $S^{(j)}$ and $U^{(j)}$ denote the sequence $S$ and the set $U$ at the beginning of the $j$-th iteration of the main loop of Algorithm~\ref{alg:TBS_simplified}.
We first show the following:

\begin{lemma}\label{le:monotonicity_of_u}
$U^{(j+1)}=U^{(j)}\setminus\{v: v$ became white and isolated while toggling nodes in line~\ref{li:toggle_v} of $j$-th iteration of the main loop of Algorithm~\ref{alg:TBS_simplified}$\}$. 
\end{lemma}
\begin{proof}
By Lemma~\ref{le:properties}(\ref{it:prop_subgr_o},\ref{it:prop_s2_not_on_u}) and Claim~\ref{cl:neighborhood_of_vmid}, graphs $G/S^{(j)}$ and $G/S^{(j+1)}$ differ only on some nodes from $U^{(j)}$, that is only nodes from $U^{(j)}$ can have different colors and both graphs have exactly the same edges except possibly of the edges that have both endpoints in $U^{(j)}$.
Consider a node $v\in U^{(j)}$ toggled in line~\ref{li:toggle_v}.
All nodes that became white and isolated in $G_1$ while toggling $v$ are white and isolated in $G/S^{(j+1)}$, so they do not belong to $U^{(j+1)}$ and the claim follows.
\end{proof}

Similarly, let $S_{2}^{(j)}$ denote the sequence $S_{2}$ in the $j$-th iteration of the main loop of Algorithm~\ref{alg:TBS_simplified}. We show:

\begin{lemma}\label{le:monotonicity_of_s}
$S_2^{(j)}$ is a (not necessarily proper) suffix of~$S_2^{(j+1)}$.
\end{lemma}
\begin{proof}
Consider the $j$-th iteration of the main while loop and let $S_3'$ be the sequence $S_3$ when reaching line~\ref{li:check}, $S_3''$ be the sequence $S_3$ when reaching line~\ref{li:set_new_S}.
Let $U^{(j)}$ and $U^{(j+1)}$ be defined as in Lemma~\ref{le:monotonicity_of_u}.

By the condition in line~\ref{li:loop_add_nodes_to_s3}, there are no black nodes from $U^{(j)}$ in $G_1/S_3'$,
and by Lemma~\ref{le:Lemma1} there are no edges between the nodes of $U \setminus \NN (v_{mid})$ and the nodes outside $U^{(j)}$.
Then, by Claim~\ref{cl:neighborhood_of_vmid}, there are no black nodes from $U^{(j)}$ in $G_1/S_3''/v_{mid}$,
and further there are no edges between the nodes of $U^{(j)}$ and the nodes outside $U^{(j)}$ there.
Thus, no further toggles can affect the nodes from $U^{(j)}$.
We conclude that, at the end of the $j$-th iteration, for every $i=2,3,\ldots$, there are no black nodes in the connected component
containing the nodes from $U^{(j)}$ in $G_1/S_3''/v_{mid}/S_{2}^{(j)}[2\ldots i]$.
However, by the previous lemma $U^{(j+1)}\subseteq U^{(j)}$, so this in particular holds for the nodes from $U^{(j+1)}$.
Hence, the split in the $j$-th iteration must be before $v_{mid},S_{2}^{(j)}[2\ldots i]=S_{2}^{(j)}$.
\end{proof}

This allows us to keep track of sequence $S$ split into $S_1$ and $S_2$ and either move elements from the end of $S_1$ to the beginning of $S_2$ or add nodes at the end of $S_1$.
With this observation, we transform Algorithm~\ref{alg:TBS_simplified} to Algorithm~\ref{alg:TBS_faster}.
Its correctness follows from the correspondence between appropriate variables in subsequent iterations of the main while loops:

\begin{lemma}
 After the loop in lines~\ref{li:loop_undo_moves}-\ref{li:loop_undo_moves_end} in the $j$-th iteration of the main loop of Algorithm~\ref{alg:TBS_faster} and after lines~\ref{li:define_U}-\ref{li:set_s3_empty} in the $j$-th iteration of the main loop of Algorithm~\ref{alg:TBS_simplified} the following are equal:
 \begin{itemize}
  \item sequence $S_1$ in Algorithm~\ref{alg:TBS_faster} and sequence $S_1$ in Algorithm~\ref{alg:TBS_simplified},
  \item sequence $S_2$ in Algorithm~\ref{alg:TBS_faster} and sequence $S_2$ in Algorithm~\ref{alg:TBS_simplified},
  \item graph $G$ in Algorithm~\ref{alg:TBS_faster} and graph $G_1$ in Algorithm~\ref{alg:TBS_simplified},
  \item set $V$ in Algorithm~\ref{alg:TBS_faster} and set $U$ in Algorithm~\ref{alg:TBS_simplified},
 \end{itemize}

After the lines~\ref{li:toggle_black_nodes_loop}-\ref{li:undo_last_move} in the $j$-th iteration of the main loop of Algorithm~\ref{alg:TBS_faster} and after lines~\ref{li:loop_add_nodes_to_s3}-\ref{li:undo_toggle_w} in the $j$-th iteration of the main loop of Algorithm~\ref{alg:TBS_simplified} the following are equal:
  \begin{itemize}
  \item sequence $S_1$ in Algorithm~\ref{alg:TBS_faster} and sequence $S_1,S_3$ in Algorithm~\ref{alg:TBS_simplified},
  \item sequence $S_2$ in Algorithm~\ref{alg:TBS_faster} and sequence $S_2$ in Algorithm~\ref{alg:TBS_simplified},
  \item graph $G$ in Algorithm~\ref{alg:TBS_faster} and graph $G_1$ in Algorithm~\ref{alg:TBS_simplified},
  \item set $V$ in Algorithm~\ref{alg:TBS_faster} and set $U$ in Algorithm~\ref{alg:TBS_simplified},
 \end{itemize} 
\end{lemma}
\begin{proof}
We proceed by induction on the number of iterations.
The invariants on $S_{1},S_{2},S_{3}$ hold by the monotonicity shown in Lemma~\ref{le:monotonicity_of_s}.
The invariant on $G$ and $G_{1}$ holds trivially, because we always do or undo the corresponding toggles.
Finally, to relate $V$ and $U$ we observe that the invariant initially holds, because we have removed from $V$
all nodes that became isolated and white in lines~\ref{li:greedily_extend}-\ref{li:greedily_extend_end} of Algorithm~\ref{alg:TBS_faster},
and there are no black nodes in $G/S_{1}$ by condition of the while loop. 
Then, recall that we have already characterised how the set $U$ changes in
subsequent iterations in Lemma~\ref{le:monotonicity_of_u}, and this is exactly how $V$ is being updated.
\end{proof}

\begin{algorithm}[h]
\begin{algorithmic}[1]
\Function{Process}{graph $G$ with no non-singleton all-white connected components}
\State $S_1,S_2:=()$
\State $V:=$ the set of all nodes of $G$ that are non-isolated or black

\Statex
\Function{\ToggleBlackNode}{}()
\State  $v:=$ a black node from $V$
\State  $\toggle(v)$
\State  remove from $V$ all nodes that became isolated and white (in particular: node $v$)
\State  $S_1:= S_1,v$
\EndFunction
\Statex
\Function{\UndoLastMove}{}()
\State  $w:=$ last element of $S_1$
\State  remove last element (that is: $w$) from $S_1$
\State  undo $\toggle(w)$
\State  \Return w
\EndFunction 
\Statex
\While{there is a black node in V}\label{li:greedily_extend}
\State \textsc{\ToggleBlackNode}()\label{li:greedily_extend_end}
\EndWhile
\While{$V$ is non-empty}\label{li:final_loop_start}
\While{there is no black node in $V$}\label{li:loop_undo_moves}
\State   $w:=$ \textsc{\UndoLastMove}()
\State   $S_2:= w,S_2$\label{li:loop_undo_moves_end}
\EndWhile
\While{there is a black node in $V$}:\label{li:toggle_black_nodes_loop}
\State   \textsc{\ToggleBlackNode}()\label{li:toggle_black_node}
\EndWhile
\If{$S_2[1]$ is white}
\State \textsc{\UndoLastMove}() \label{li:undo_last_move}
\EndIf 
\EndWhile
\State \Return $S_1,S_2$
\EndFunction
\end{algorithmic}
\caption{}
\label{alg:TBS_faster}
\end{algorithm}

Finally, rearranging the order of while loops and checks gives us exactly Algorithm~\ref{alg:TBS_orig} from~\cite{TannierBS07}.

\end{document}